\documentclass[aps,tightenlines,10pt,superscriptaddress]{revtex4}%
\usepackage{amssymb}
\usepackage{amsfonts}
\usepackage{amsmath}
\usepackage{amsmath}
\usepackage{pstricks}
\usepackage{graphicx}%
\providecommand{\U}[1]{\protect\rule{.1in}{.1in}}
\newtheorem{theorem}{Theorem}

\newtheorem{corollary}[theorem]{Corollary}

\newtheorem{definition}[theorem]{Definition}
\newtheorem{example}[theorem]{Example}

\newtheorem{lemma}[theorem]{Lemma}

\newtheorem{problem}[theorem]{Problem}
\newtheorem{proposition}[theorem]{Proposition}
\newtheorem{remark}[theorem]{Remark}

\newenvironment{proof}[1][Proof]{\noindent\textbf{#1.} }{\ \rule{0.5em}{0.5em}}
\begin{document}
\title{A notion of graph likelihood and an infinite monkey theorem}
\author{Christopher Banerji}
\affiliation{Department of Computer Science, and Centre of Mathematics \& Physics in the
Life Sciences and Experimental Biology, University College London, London WC1E
6BT, United Kingdom}
\affiliation{Statistical Cancer Genomics, Paul O'Gorman Building, UCL Cancer Institute,
University College London, London WC1E 6BT, United Kingdom}
\author{Toufik Mansour}
\affiliation{Department of Mathematics, University of Haifa, Haifa 31905, Israel}
\author{Simone Severini}
\affiliation{Department of Computer Science, and Department of Physics \& Astronomy,
University College London, WC1E 6BT London, United Kingdom}

\begin{abstract}
We play with a graph-theoretic analogue of the folklore infinite monkey
theorem. We define a notion of graph likelihood as the probability that a
given graph is constructed by a monkey in a number of time steps equal to the
number of vertices. We present an algorithm to compute this graph invariant
and closed formulas for some infinite classes. We have to leave the computational
complexity of the likelihood as an open problem.

\end{abstract}
\maketitle

\section{Introduction}

The \emph{infinite monkey theorem} is part of the popular culture \cite{wi}. A
monkey sits in front to a typewriter hitting random keys. The probability that
the monkey will type any given text tends to one, as the amount of time the
monkey spends on the typewriter tends to infinity. The usual example is of
Shakespeare's Hamlet. Of course, the term \textquotedblleft
monkey\textquotedblright\ can refer to some abstract device producing random
strings of symbols (\emph{e.g.}, zeros and ones).

In this note, we consider an infinite monkey theorem, but for graphs rather than strings. Our setting involves a device which performs
\textquotedblleft non-preferential\textquotedblright\ attachment \cite{js12}.
At time step $t+1$, a new vertex is added to a graph $G_{t}$ -- the process
starts from the single vertex graph, $G_{1}$. The degree and the neighbours of
the newly added vertex at step $t+1$ are both chosen at random. The degree of
the vertex is then $k\in\{0,1,...,t\}$ and its neighbours are $k$ random
vertices in $G_{t}$. The (in fact obvious) analogue of the infinite monkey
theorem is that every graph can be constructed in this way: the probability
that the monkey will construct a given graph tends to one, as the amount of
time the monkey spends on the \textquotedblleft graphwriter\textquotedblright%
\ tends to infinity. Notice that the monkey makes two random choices, but
these can be seen as a single one. Also, notice that the theorem is indeed a
corollary of the usual infinite monkey theorem since we could encode a graph
in a string (for example, by vectorizing the adjacency matrix). Below is a
monkey enjoying the construction of the Petersen graph: %

\begin{figure}[h]
\centering
\includegraphics[height=3.3497in,
width=2.3383in]{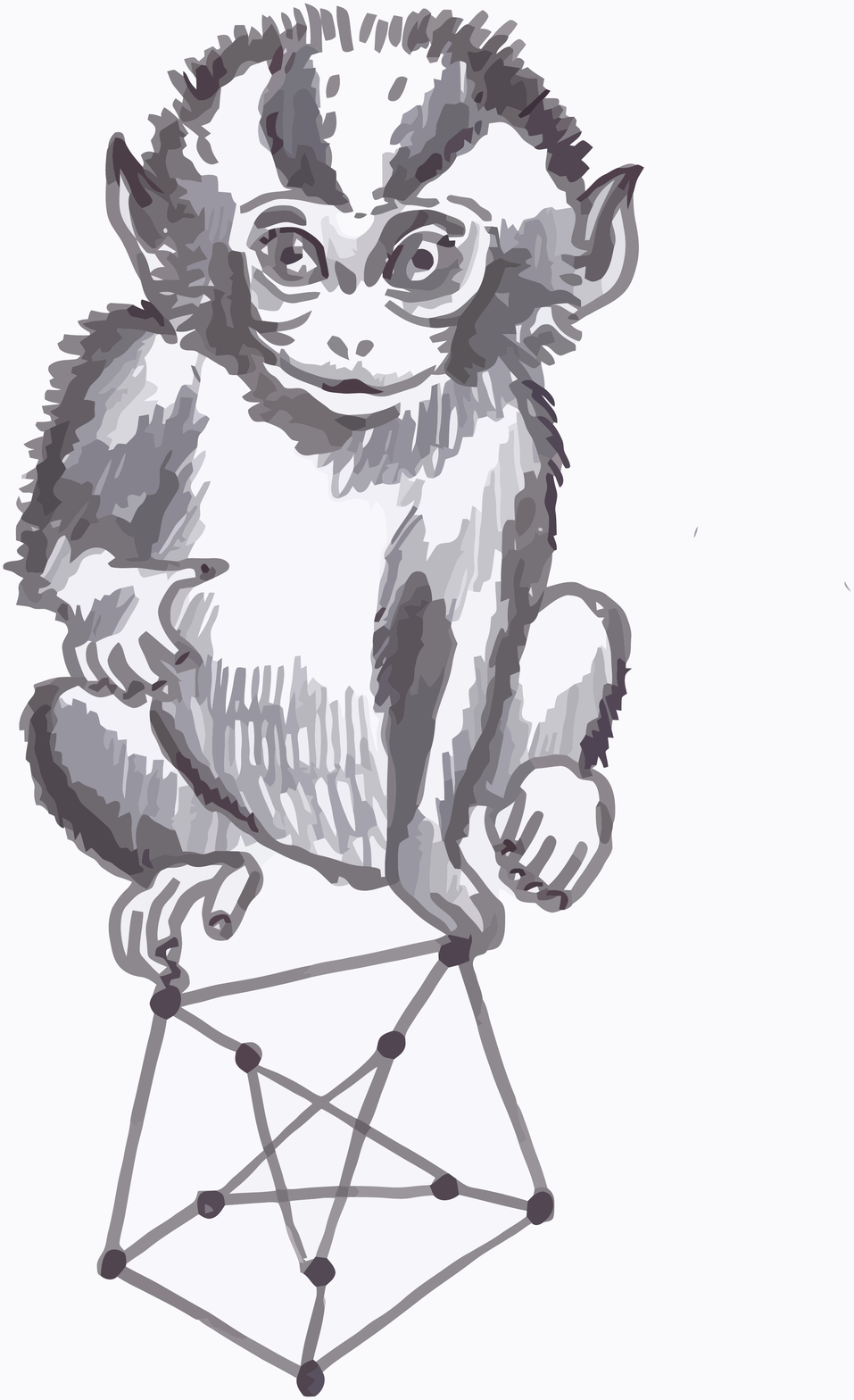}
\end{figure}

The construction is basically an excuse to discuss a graph invariant which we
call \emph{(graph) likelihood}. This is the probability that a given graph on
$t$ vertices is obtained by the construction after exactly $t$ steps. In other
words, this is the probability that a monkey constructs a given graph on $t$
vertices in exactly $t$ seconds, assuming that the monkey adds a new vertex
each second. For a string, this would correspond to the probability that the
monkey types a given text in a time equal to the length of the string produced.

The likelihood is a plausible measure to quantify how difficult is to
construct a graph in the way we propose. Intuitively, graphs with more
symmetries have generally smaller likelihood. We will show as expected that
bounds on the likelihood can be given in terms of the automorphism group.
Specifically, the likelihood can not be larger than the reciprocal of the
number of automorphisms. Graphs with trivial automorphism group are then
potentially the ones admitting highest likelihood. We will describe an
algorithm to compute the likelihood of a given graph. The algorithm uses a
rooted tree decomposition which takes into account all possible ways to
construct the graph by adding one vertex at the time. The algorithm suggests a
closed formula for the likelihood.

The reminder of the paper is organized as follows. In Section II, we define
the likelihood. In Section III, we give closed formulas for complete graphs,
star graphs, paths, and cycles. In Section IV, we describe the algorithm.
Section V lists some open problems. In particular, we could not to prove the
complexity of computing the likelihood. The paper is practically
self-contained.

\section{Graph likelihood}

As usual, $G=(V,E)$ denotes a \emph{graph}: $V(G)=\{v_{1},v_{2},...,v_{t}\}$
is a set whose elements are called \emph{vertices} and $E(G)\subseteq
V(G)\times V(G)-\{\{v_{i},v_{i}\}:v_{i}\in V(G)\}$ is a set whose elements are
called \emph{edges}. The graph with a single vertex and no edges is denoted by
$K_{1}$. Our main object of study will be the construction given in the
following definition. This is a special case of a construction already
presented in \cite{js12}.

\begin{definition}
[Construction]\label{def1}We construct a graph $G_{t}=(V,E)$, starting from
$G_{1}=K_{1}$. The construction involves an iteration with discrete steps. At
the $t$-th step of the iteration, the graph $G_{t-1}$ is transformed into the
graph $G_{t}$. The $t$-th step of the iteration is divided into three substeps:

\begin{enumerate}
\item We select a number $k\in\{0,1,...,t-1\}$ with equal probability.

Assume that we have selected $k$.

\item We select $k$ vertices of $G_{t-1}$ with equal probability.

Assume that we have selected the vertices $v_{1},v_{2},...,v_{k}\in
V(G_{t-1})$.

\item We add a new vertex $t$ to $G_{t-1}$ and the edges $\{v_{1}%
,t\},\{v_{2},t\},...,\{v_{k},t\}\in E(G_{t})$.
\end{enumerate}
\end{definition}

On the basis of the construction, the following definition is natural:

\begin{definition}
[Graph likelihood]Let $G$ be a graph on $t$ vertices. The (\emph{graph}%
)\ \emph{likelihood} of $G$, denoted by $\mathcal{L}(G)$, is defined as the
probability that $G_{t}=G$, where $G_{t}$ is the graph given by the
construction in Definition \ref{def1}:
\[
\mathcal{L}(G):=\emph{Pr}[G_{t}=G].
\]

\end{definition}

For clarifying this notion, in the next section, we write closed formulas for
the likelihood of graphs in some infinite families. We use rather
uninteresting proof techniques, but these serve the purpose, at least for very
simple graphs.

\section{Examples}

The \emph{complete graph} $K_{t}$ is defined as the graph on $t$ vertices and
$t(t-1)/2$ edges.

\begin{proposition}
Let $K_{t}$ be the complete graph on $t$ vertices. Then, $\mathcal{L}%
(K_{t})=1/t!$.
\end{proposition}

\begin{proof}
For $K_{t}$, the only significant step of the construction is the first one
(\emph{i.e.}, the selection of a number $k\in\{0,1,...,t-1\}$ with equal
probability). Therefore, $\mathcal{L}(K_{t})=\prod_{i=1}^{t}\frac{1}{i}$. This
equals $1/t!$ by definition.
\end{proof}

\bigskip

The \emph{star graph} $K_{1,t-1}$ is defined as the graph on $t$ vertices,
$v_{1},v_{2},...,v_{t}$, and the edges $\{v_{1},v_{2}\},\{v_{1},v_{3}%
\},...,\{v_{1},v_{t-1}\}$. In a graph $G=(V,E)$, the \emph{degree} of a vertex
$i\in V(G)$ is defined and denoted by $d(i)=\left\vert \{j:\{i,j\}\in
E(G)\right\vert $.

\begin{proposition}
Let $K_{1,t-1}$ be the star graph on $t$ vertices. Then, $\mathcal{L}%
(K_{1,t-1})=\frac{t}{(t!)^{2}}\sum_{i=0}^{t-1}i!$.
\end{proposition}

\begin{proof}
The star graph $K_{1,t-1}$ has $1$ vertex of degree $t-1$ and $t-1$ vertices
of degree $1$. There are three cases relevant to the construction of
$K_{1,t-1}$ such that $G_{t}=K_{1,t-1}$:

\begin{enumerate}
\item Suppose we add $t-1$ vertices, $1,2,...,t-1$, of degree $0$. At time
$t$, we add a vertex, $t$, of degree $t-1$. Since $\Pr[d(i)=0]=1/i$, for
$i=1,2,\ldots,t-1$, and $\Pr[d(t)=t-1]=1/t$, $\Pr[G_{t}=K_{1,t-1}$ by
(1)$]=\prod_{i=1}^{t-1}\Pr[d(i)=0]\cdot\Pr[d(t)=t-1]=\left(  \prod_{i=1}%
^{t-1}\frac{1}{i}\right)  \cdot\frac{1}{t}=\frac{1}{t!}$.

\item Suppose there is an edge $\{1,2\}\in G_{2}$. Since $\Pr[d(3)=1]=1/3$, we
distinguish two cases:

\begin{enumerate}
\item $\Pr[\{1,3\}\in E(G_{3})]=1/2$: $\Pr[d(2)=1]\cdot\Pr[d(3)=1]\cdot
\Pr[\{1,3\}\in E(G_{3})]=\frac{1}{2}\cdot\frac{1}{3}\cdot\frac{1}{2}=\frac
{1}{12}$. If $\{1,3\}\in E(G_{3})$ then each other edge of $G_{t}$, with
$t\geq4$, must be of the form $\{1,4\},\ldots,\{1,t\}$ and $\Pr[\{1,t\}\in
E(G_{3})]=\frac{1}{t}\cdot\frac{1}{t-1}$.

\item $\Pr[\{2,3\}\in E(G_{3})]=1/2$: $\Pr[d(2)=1]\cdot\Pr[d(3)=1]\cdot
\Pr[\{2,3\}\in E(G_{3})]=\frac{1}{2}\cdot\frac{1}{3}\cdot\frac{1}{2}=\frac
{1}{12}$. If $\{2,3\}\in E(G_{3})$ then the situation is analogous to the
previous case.
\end{enumerate}

By combining together (a) and (b), it follows that
\[
\Pr[G_{t}=K_{1,t-1}\text{ by (2)}]=2\prod_{i=2}^{t}\left(  \frac{1}{i}%
\cdot\frac{1}{i-1}\right)  =\frac{2}{t!(t-1)!}.
\]

\item Suppose we add $k-1$ vertices, $1,2,\ldots,k-1$, of degree $0$, where
$k\geq3$. At time $k$, we add a vertex, $k$, of degree $k-1$. Since
$\Pr[d(i)=0]=1/i$, for $i=1,2,\ldots,k-1$, and $\Pr[d(k)=k-1]=1/k$,
\[
\Pr[G_{k}=K_{1,k-1}\text{ by (3)}]=\prod_{i=1}^{t-1}\Pr[d(i)=0]\cdot
\Pr[d(k)=k-1]=\left(  \prod_{i=1}^{k-1}\frac{1}{i}\right)  \cdot\frac{1}%
{k}=\frac{1}{k!}.
\]
The remaining $t-k$ vertices, $k+1,k+2,\ldots,t$, must be of the form
$\{k,k+1\},\{k,k+2\},\ldots,\{k,t\}$ and
\[
\Pr[\{k,k+j\}\in E(G_{k+j})]=\frac{1}{k+j}\cdot\frac{1}{k+j-1},
\]
for each $j=1,2,\ldots,t-k$. Hence,
\[
\Pr[G_{t}=K_{1,t-1}\text{ by (3)}]=\sum_{k=3}^{t-1}\frac{1}{k!}\prod
_{i=k+1}^{t}\left(  \frac{1}{i}\cdot\frac{1}{i-1}\right)  .
\]

\end{enumerate}

The analysis carried out with the three cases above is sufficient to obtain
the following formula:
\begin{align*}
\mathcal{L}(K_{1,t-1})  &  =\Pr[G_{t}=K_{1,t-1}\text{ by (1)}]+\Pr
[G_{t}=K_{1,t-1}\text{ by (3)}]+\Pr[G_{t}=K_{1,t-1}\text{ by (2)}]\\
&  =\frac{1}{t!}+\frac{2}{t!(t-1)!}+\sum_{k=3}^{t-1}\frac{1}{k!}\prod
_{i=k+1}^{t}\frac{1}{i(i-1)}=\frac{t}{(t!)^{2}}\sum_{i=0}^{t-1}i!.
\end{align*}

\end{proof}

\section{Computation of the likelihood}

Is the likelihood defined for any graph? The answer is \textquotedblleft
yes\textquotedblright, as demonstrated by the next statement. This is a
plausible graph-theoretic analogue of the infinite monkey theorem:

\begin{proposition}
\label{procon}Any graph can be obtained with the construction in Definition
\ref{def1}.
\end{proposition}

\begin{proof}
An \emph{orientation} of $G$ is a function $\alpha:E(G)\longrightarrow
E^{+}(G)$, where $E^{+}(G)$ is a set whose elements, called \emph{arcs}, are
ordered pairs of vertices such that either $\alpha(\{i,j\})=(i,j)$ or
$\alpha(\{i,j\})=(j,i)$, for each $\{i,j\}\in E(G)$. An orientation is
\emph{acyclic} if it does not contain any directed cycles, \emph{i.e.},
distinct vertices $v_{1},...,v_{k}$ such that $(v_{1},v_{2}),(v_{2}%
,v_{3}),...,(v_{k-1},v_{k}),(v_{k},v_{1})$ are arcs. Clearly, every graph has
an acyclic orientation. Every acyclic orientation determines at least one
linear ordering $v_{1}<v_{2}<\cdots<v_{n}$ of the vertices such that, for each
edge $\{v_{i},v_{j}\}$, we have $\alpha(\{v_{i},v_{j}\})=(v_{i},v_{j})$ if and
only if $v_{i}<v_{j}$. This is also called a \emph{topological ordering} of
the vertices relative to the orientation. For a graph $G$, let $V(G)=\{w_{1}%
,w_{2},...,w_{t}\}$ and let $w_{1}<w_{2}<\cdots<w_{t}$ realize a topological
ordering. We can always obtain $G_{t}=G$, if in the iteration we have
$v_{1}=w_{1},v_{2}=w_{2},...,v_{t}=w_{t}$.
\end{proof}

\bigskip

And, of course:

\begin{proposition}
Every graph on $n$ vertices has a positive likelihood. (More formally,
$\mathcal{L}(G)>0$ for any graph $G$.)
\end{proposition}

Proposition \ref{procon} suggests a natural computational problem:

\begin{problem}
[Likelihood computation]\label{pro1}\textbf{Given:} A graph $G$.
\textbf{Task:} Compute $\mathcal{L}(G)$.
\end{problem}

There are surely many ways to approach this problem. We consider an algorithm
based on a tree whose vertices represent all intermediate graphs obtained
during the construction.

\begin{definition}
[Identity representation]Let $G=(V,E)$ be a graph on the set of vertices
$V(G)=\{v_{1},v_{2},...,v_{t}\}$. Let us fix an arbitrary labeling of the
vertices of $G$ by a bijection $f:V(G)\longrightarrow\{1,2,...,t\}$. Once
fixed the bijection, let us label the first row (resp. column)\ of the
adjacency matrix of $G$, $A(G)$, by the number $t$, the second one by
$t-1$,..., the last one by $1$. The bijection $f$ can be then represented by
the ordered set $(1,2,...,t)$. We can then define an acyclic orientation of
the edges such that $\alpha(\{i,j\})=(i,j)$ if and only if $i<j$, with
$i,j=1,2,...,t$. The topological ordering relative to the orientation defines
$G_{1}=(1,\emptyset)$, $G_{2}=\{\{1,2\},E(G_{2})\}$,...,$G_{t}%
=\{\{1,2,...,t\},E(G_{t}))=G$. The pair $(A(G),(1,2,...,t))$ given by the
adjacency matrix $A(G)$ together with the ordered set \emph{id}
$:=(1,2,...,t)$ is said to be the \emph{identity representation} of $G$.
\end{definition}

\begin{remark}
The identity representation is arbitrary, since it entirely depends on the
bijection $f$.
\end{remark}

A \emph{permutation of length} $n$ is a bijection
$p:\{1,2,...,t\}\longrightarrow\{1,2,...,t\}$. Hence, each permutation $p$
corresponds to an ordered set $(p(1),p(2),...,p(t))$. The set of all
permutations of length $t$ is denoted by $S_{t}$. A \emph{permutation matrix}
$P$ induced by a permutation $p$ of length $t$ is an $t\times t$ matrix such
that $[P]_{i,j}=1 $ if $p(i)=j$ and $[P]_{i,j}=0$, otherwise. Lower case
letters denote permutations; upper case letters their induced matrices.

\begin{definition}
[(Generic) Representation]Let $G=(V,E)$ be a graph on $t$ vertices. Let
$(A(G),$ \emph{id}$)$ be the identity representation of $G$. The pair
$(PA(G)P^{T},p)$, where $P$ is a permutation matrix induced by the permutation
$p$ is said to be a \emph{representation} of $G$. A representation
$(PA(G)P^{T},p)$ is also denoted by $A_{p}(G)$.
\end{definition}

An \emph{automorphism} of a graph $G=(V,E)$ is a permutation
$p:V(G)\longrightarrow V(G)$ such that $\{v_{i},v_{j}\}\in E(G)$ if and only
if $\{p(v_{i}),p(v_{j})\}\in E(G)$. The set of all automorphisms of $G$, with
the operation of composition of permutations \textquotedblleft$\circ
$\textquotedblright, is a permutation group denoted by Aut$(G)$. Such a group
is the \emph{full automorphism group} of $G$. The permutation matrices $P$,
induced by the elements of Aut$(G)$, are precisely the matrices such that
$PA(G)P^{T}=A(G)$, \emph{i.e.}, $PA(G)=A(G)P$.

\begin{lemma}
Let $G=(V,E)$ be a graph on $t$ vertices. The total number of different
representations of $G$ is $t!/\left\vert \text{\emph{Aut}}(G)\right\vert $.
\end{lemma}

\begin{proof}
Let $A_{\text{id}}(G)$ be an identity representation of $G$. By the definition
of full automorphism group, for each permutation $p\in$ Aut$(G)$, we have
$A_{p}(G)=PA_{\text{id}}(G)P^{T}=A_{\text{id}}(G)$. Let $q\in S_{t}-$
Aut$(G)$. Then, there is a unique permutation $r\in S_{t}-$ Aut$(G)$ such that
$q=p\circ r$. It follows that $QA_{\text{id}}(G)Q^{T}=PRA_{\text{id}}%
(G)R^{T}P^{T}=PA_{r}(G)P^{T}=A_{r}(G)$. This indicates that each
representation of $G$ belongs to an equivalence class of representations.
Since $\left\vert S_{t}\right\vert =t!$, the total number of different
representations of $G$, \emph{i.e.}, the total number of equivalence classes
of representations, is $t!/\left\vert \text{Aut}(G)\right\vert $.
\end{proof}

\begin{remark}
In the language of elementary group theory, the equivalence classes are the
(left)\ cosets of the subgroup \emph{Aut}$(G)$ in $S_{t}$.
\end{remark}

In order to design an algorithm for $\mathcal{L}(G)$, we need some further
definitions. A \emph{subgraph} $H=(V^{\prime},E^{\prime})$ of a graph
$G=(V,E)$ is a graph such that $V^{\prime}\subseteq V$ and $E^{\prime
}\subseteq E$. We say that a graph $G$ \emph{contains} a graph $H$ if there is
a subgraph of $G$ isomorphic to $H$.

\begin{definition}
Let $G$ be any nonempty graph with $t$ vertices, a \emph{path construction} of
$G$ is a sequence $(H_{1},H_{2},\ldots,H_{t})$ of $t$ graphs such that $H_{i}$
has $i$ vertices, $i=1,2,\ldots,t$, and $H_{i}\subset H_{i+1}$, for each
$i=1,2,\ldots,t-1$; moreover, $H_{t}\cong G$. We denote the set of all path
constructions of a graph $G$ by Path$(G)$.
\end{definition}

It is clear that each path construction corresponds to an equivalence class of representations.

The set of path constructions can be represented as a rooted tree $T_{G}$ as follows:

\begin{itemize}
\item The root of $T_{G}$ is $T_{1}$. This is the empty graph with a single vertex.

\item Assume we already have all the vertices at level $i$ (the level of the
root is taken to be $1$) in the tree $T_{G}$. Let $(T_{1},T_{2},\ldots,T_{i})$
be a path in $T_{G}$, if there exists a path construction $L=(T_{1}%
,T_{2},\ldots,T_{i},H_{i+1},\ldots,H_{t})\in$ Path$(G)$ then we define
$H_{i+1}$ to be one of the children of the node $T_{i}$ in $T_{G}$.
\end{itemize}

\begin{example}
\label{ext1}The rooted tree $T_{P_{3}}$ is given by \begin{figure}[pth]
\begin{center}
\begin{pspicture}(0,-1)(4.4,2.6)
\pscircle*(1.5,2.2){0.08}\put(1.35,2.4){\small$v_1$}\put(0.6,2.2){\small$(T_1)$}
\psline(1.5,2)(.2,1)\psline(1.5,2)(2.5,1)
\pscircle*(-.2,0.55){0.08}\put(-.35,.7){\small$v_1$}\pscircle*(.5,0.55){0.08}\put(.35,.7){\small$v_2$}\put(-1.3,.6){\small$(T_{21})$}
\put(2.5,0){\pscircle*(-.2,0.55){0.08}\put(-.35,.7){\small$v_1$}\pscircle*(.5,0.55){0.08}\put(.35,.7){\small$v_2$}\psline(-.2,0.55)(.5,.55)
\put(1,.6){\small$(T_{22})$}}
\psline(.2,.3)(.2,-.6)\psline(2.6,.3)(2.3,-.5)\psline(2.6,.3)(4.3,-.5)
\put(-.5,-1.5){\pscircle*(-.2,0.55){0.08}\put(-.35,.7){\small$v_1$}\pscircle*(.5,0.55){0.08}\put(.35,.7){\small$v_3$}
\pscircle*(1.2,0.55){0.08}\put(1.05,.7){\small$v_2$}\psline(1.2,0.55)(.5,.55)\psline(-.2,0.55)(.5,.55)\put(.1,.1){\small$(T_{31})$}}
\put(1.8,-1.5){\pscircle*(-.2,0.55){0.08}\put(-.35,.7){\small$v_1$}\pscircle*(.5,0.55){0.08}\put(.35,.7){\small$v_2$}
\pscircle*(1.2,0.55){0.08}\put(1.05,.7){\small$v_3$}\psline(1.2,0.55)(.5,.55)\psline(-.2,0.55)(.5,.55)\put(.1,.1){\small$(T_{32})$}}
\put(4,-1.5){\pscircle*(-.2,0.55){0.08}\put(-.35,.7){\small$v_3$}\pscircle*(.5,0.55){0.08}\put(.35,.7){\small$v_1$}
\pscircle*(1.2,0.55){0.08}\put(1.05,.7){\small$v_2$}\psline(1.2,0.55)(.5,.55)\psline(-.2,0.55)(.5,.55)\put(.1,.1){\small$(T_{33})$}}
\end{pspicture}
\end{center}
\end{figure}

The above figure shows that the set of path constructions of $P_{3}$ is given
by
\[
\emph{Path}(P_{3})=\{(T_{1},T_{21},T_{31}),(T_{1},T_{22},T_{32}),(T_{1}%
,T_{22},T_{33})\}.
\]

\end{example}

Let $P=(H_{1},H_{2},\ldots,H_{t})\in$ Path$(G)$ be any path construction of
$G$. Fix $i$, then $H_{i+1}$ is obtained by adding a vertex $v_{i+1}$ of
degree $d_{i+1}(P)$ to the graph $H_{i}$. Hence
\[
\Pr[G_{t}=G,\,P\mbox{ is a path construction of }G]=\prod_{i=1}^{t}\frac
{1}{i\binom{i-1}{d_{i}(P)}}=\frac{1}{t!\prod_{i=1}^{t}\binom{i-1}{d_{i}(P)}}.
\]

From this algorithm, we obtain a relation between $\mathcal{L}(C_{n})$ and
$\mathcal{L}(P_{n})$ as follows. Recall that $C_{n}$ is the cycle on $n$
vertices and $P_{n}$ is the path on $n$ vertices.

\begin{corollary}
For all $n\geq3$, $\mathcal{L}(C_{n})=\mathcal{L}(P_{n-1})/n\binom{n-1}{2}$.
\end{corollary}

An algorithm for computing $\mathcal{L}(G)$ can be based on the following theorem:

\begin{theorem}
\label{thflg} Let $G$ be a graph on $t$ vertices. Then
\[
\mathcal{L}(G)=\sum_{P\in\text{ \emph{Path}}(G)}\frac{1}{t!\prod_{i=1}%
^{t}\binom{i-1}{d_{i}(P)}}.
\]

\end{theorem}

A simple example is useful:

\begin{example}
Let $P_{3}$ be the path graph on $3$ vertices. By Example \ref{ext1}, we find
that
\begin{align*}
\Pr[G_{t} &  =G,\,(T_{1},T_{21},T_{31}%
)\mbox{ is a path construction of }G]=1\cdot\frac{1}{2}\cdot\frac{1}{3}%
=\frac{1}{6},\\
\Pr[G_{t} &  =G,\,(T_{1},T_{22},T_{32}%
)\mbox{ is a path construction of }G]=1\cdot\frac{1}{2}\cdot\frac{1}{3\cdot
2}=\frac{1}{12},\\
\Pr[G_{t} &  =G,\,(T_{1},T_{22},T_{33}%
)\mbox{ is a path construction of }G]=1\cdot\frac{1}{2}\cdot\frac{1}{3\cdot
2}=\frac{1}{12}.
\end{align*}
Then, $\mathcal{L}(G)=\frac{1}{6}+\frac{1}{12}+\frac{1}{12}=\frac{1}{3}$.
\end{example}

By Theorem \ref{thflg} and the fact that $|$Path$(G)|$ is exactly equal to the
number of representations of $G$, \emph{i.e.} $|$Path$(G)|=t!/\left\vert
\text{Aut}(G)\right\vert $, we obtain the following bounds:

\begin{corollary}
Let $G$ be any nonempty graph on $t$ vertices. Then
\[
\frac{1}{|\text{\emph{Aut}}(G)|\prod_{i=1}^{t}\binom{i-1}{\lfloor
(i-1)/2\rfloor}}\leq\mathcal{L}(G)\leq\frac{1}{|\emph{Aut}(G)|}.
\]

\end{corollary}

We give two general examples:

\begin{example}
Let $G$ be a graph on $t$ vertices with exactly $s$ edges incident with $2s$
vertices. Any path $P$ of Path$(G)$ can be seen as a path from a single vertex
to the graph $G$. At levels $i_{1},i_{2},\ldots,i_{s}$, we have added an edge
between the new vertex and a vertex of degree zero. In all other levels we
just added a new vertex. Therefore,
\[
\mathcal{L}(G)=\frac{1}{t!}\sum_{2\leq i_{1}<i_{2}<\cdots<i_{s}\leq t}%
\prod_{j=1}^{s}\frac{i_{j}+1-2j}{i_{j}-1}.
\]

\end{example}

\begin{example}
Let $G$ be a graph on $t$ vertices with exactly one edge, then
\[
\mathcal{L}(G)=\frac{1}{t!}\sum_{i=2}^{t}1=\frac{t-1}{t!}.
\]

\end{example}

\begin{example}
Let $G$ be a graph on $t$ vertices with exactly two edges incident on four
vertices (a matching with two edges), then
\[
\mathcal{L}(G)=\frac{1}{t!}\sum_{i=2}^{t}\left(  \frac{i-2}{i}+\frac{i-1}%
{i+1}+\cdots+\frac{n-3}{n-1}\right)  .
\]

\end{example}

\begin{figure}[pth]
\begin{center}
\begin{pspicture}(10,4)
\put(0,4){\pscircle*(0,0){0.06}\put(0.05,0){$1$}}
\put(3,4){\pscircle*(0,0){0.06}\pscircle*(0.5,0){0.06}\put(0.55,0){$\frac{1}{2}$}}
\put(6,4){\pscircle*(0,0){0.06}\pscircle*(0.5,0){0.06}\psline(0,0)(.5,0)\put(0.55,0){$\frac{1}{2}$}}
\put(9,4){\pscircle*(0,0){0.06}\pscircle*(0.5,0){0.06}\pscircle*(1,0){0.06}\put(1.05,0){$\frac{1}{6}$}}
\put(0,3){\pscircle*(0,0){0.06}\pscircle*(0.5,0){0.06}\pscircle*(1,0){0.06}\psline(0,0)(.5,0)\put(1.05,0){$\frac{1}{3}$}}
\put(3,3){\pscircle*(0,0){0.06}\pscircle*(0.5,0){0.06}\pscircle*(1,0){0.06}\psline(0,0)(1,0)\put(1.05,0){$\frac{1}{3}$}}
\put(6,3){\pscircle*(0,0){0.06}\pscircle*(0.5,.25){0.06}\pscircle*(1,0){0.06}\psline(0,0)(.5,.25)(1,0)(0,0)\put(1.05,0){$\frac{1}{6}$}}
\put(9,3){\pscircle*(0,0){0.06}\pscircle*(0.5,0){0.06}\pscircle*(1,0){0.06}\pscircle*(1.5,0){0.06}\put(1.55,0){$\frac{1}{24}$}}
\put(0,2){\pscircle*(0,0){0.06}\pscircle*(0.5,0){0.06}\pscircle*(1,0){0.06}\pscircle*(1.5,0){0.06}\psline(0,0)(.5,0)\put(1.55,0){$\frac{1}{8}$}}
\put(3,2){\pscircle*(0,0){0.06}\pscircle*(0.5,0){0.06}\pscircle*(1,0){0.06}\pscircle*(1.5,0){0.06}\psline(0,0)(.5,0)\psline(1,0)(1.5,0)\put(1.55,0){$\frac{1}{36}$}}
\put(6,2){\pscircle*(0,0){0.06}\pscircle*(0.5,0){0.06}\pscircle*(1,0){0.06}\pscircle*(1.5,0){0.06}\psline(0,0)(1,0)\put(1.55,0){$\frac{13}{72}$}}
\put(9,2){\pscircle*(0,0){0.06}\pscircle*(0.5,0){0.06}\pscircle*(1,0){0.06}\pscircle*(1.5,0){0.06}\psline(0,0)(1.5,0)\put(1.55,0){$\frac{1}{9}$}}
\put(0,1){\pscircle*(0,0){0.06}\pscircle*(0.5,0){0.06}\pscircle*(0,.25){0.06}\pscircle*(.5,.25){0.06}\psline(0,0)(.5,0)(.5,.25)(0,.25)(0,0)\put(0.55,0){$\frac{1}{36}$}}
\put(3,1){\pscircle*(0,0){0.06}\pscircle*(0.5,0){0.06}\pscircle*(.5,.25){0.06}\pscircle*(.5,-.25){0.06}\psline(0,0)(.5,0)\psline(0,0)(.5,.25)\psline(0,0)(.5,-.25)\put(0.55,0){$\frac{5}{72}$}}
\put(6,1){\pscircle*(0,0){0.06}\pscircle*(0.5,.25){0.06}\pscircle*(1,0){0.06}\pscircle*(1.5,0){0.06}\psline(0,0)(.5,.25)(1,0)(0,0)\put(1.55,0){$\frac{5}{72}$}}
\put(9,1){\pscircle*(0,0){0.06}\pscircle*(0.5,.25){0.06}\pscircle*(1,0){0.06}\pscircle*(1.5,0){0.06}\psline(0,0)(.5,.25)(1,0)(0,0)(1.5,0)\put(1.55,0){$\frac{13}{72}$}}
\put(3,0){\pscircle*(0,0){0.06}\pscircle*(0.5,.25){0.06}\pscircle*(1,0){0.06}\pscircle*(1.5,0){0.06}\psline(0,0)(.5,.25)(1,0)(0,0)(1.5,0)(.5,.25)\put(1.55,0){$\frac{1}{8}$}}
\put(6,0){\pscircle*(0,0){0.06}\pscircle*(0.5,0){0.06}\pscircle*(0,.25){0.06}\pscircle*(.5,.25){0.06}\psline(0,0)(.5,0)(.5,.25)(0,.25)(0,0)(.5,.25)(0,.25)(.5,0)\put(0.55,0){$\frac{1}{24}$}}
\end{pspicture}
\end{center}
\caption{All non-isomorphic graphs on $t$ vertices, where $t\leq4$, and their
likelihood. }%
\end{figure}
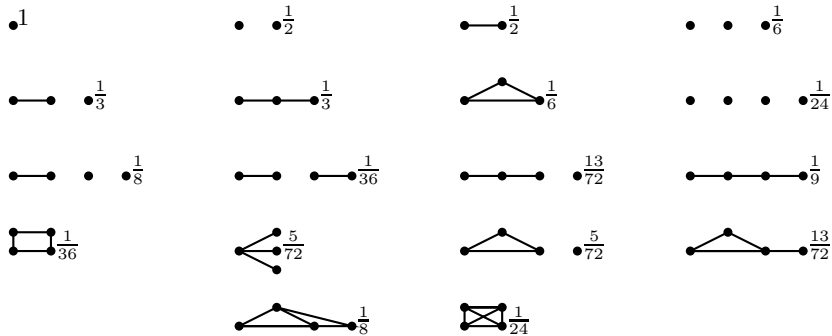

By making use of Theorem \ref{thflg}, we can prove in a straightforward way
that a graph and its complement have equal likelihood. The \emph{complement}
of a graph $G=(V,E)$, denoted by $\overline{G}$, is the graph such that
$V(\overline{G})=V(G)$ and $E(G)=V(G)\times V(G)-\{\{v_{i},v_{i}\}:v_{i}\in
V(G)\}-E(G)$.

\begin{proposition}
Let $G$ be any graph. Then $\mathcal{L}(G)=\mathcal{L}(\overline{G})$.
\end{proposition}

\section{Conclusions}

We have used a model of graph growth to introduce a notion of graph likelihood
and we have then discussed some of its basic aspects. This is the probability
that a graph is grown with the model. We have proposed an algorithm for the
computation of the likelihood, and we have bounded this graph invariant in
terms of the automorphism group. We conclude with two natural open problems:

\begin{problem}
How hard is to compute the likelihood?
\end{problem}

\begin{problem}
Which graphs are extremal with respect to the likelihood?
\end{problem}

\bigskip

\emph{Acknowledgments.} We would like to thank Ginestra Bianconi, Sebi Cioaba,
Chris Godsil, Anastasia Koroto, Matt DeVos, and Svante Janson.
%------------------------------------------------

\end{document}